\documentclass[letterpaper, 10 pt, conference]{ieeeconf}  

\usepackage{amsmath,amsfonts}

\usepackage{algorithmic}
\usepackage{algorithm}
\usepackage{array}
\usepackage{textcomp}
\usepackage{stfloats}
\usepackage{url}
\usepackage{verbatim}
\usepackage{graphicx}
\usepackage{textcomp}
\usepackage{mathrsfs} 
\usepackage{graphicx}
\usepackage{textcomp}
\usepackage{MnSymbol}
\usepackage{textcomp}
\usepackage[bbgreekl]{mathbbol}
\usepackage{multirow}
\usepackage{stfloats}
\usepackage{float}   
\usepackage{cuted}    
\usepackage{booktabs}

\usepackage[noadjust]{cite}
\DeclareMathOperator{\rank}{rank}
\DeclareMathOperator{\diag}{diag}
\DeclareMathOperator{\col}{col}

\DeclareMathOperator*{\argmin}{arg\,min}

\newtheorem{theorem}{Theorem}%
\newtheorem{lemma}{Lemma}%
\newtheorem{example}{Example}%
\newtheorem{remark}{Remark}%
\newtheorem{assumption}{Assumption}%

\IEEEoverridecommandlockouts                              

\title{\LARGE \bf
Regularized Model Predictive Control$^{*}$
}

\author{Komeil Nosrati$^{1}$, Juri Belikov$^{2}$, Aleksei Tepljakov$^{1}$, and Eduard Petlenkov$^{1}$
\thanks{$^{*}$This work was supported by the European Union’s Horizon Europe research \& innovation programme under the grant agreement No. 101120657, project ENFIELD (European Lighthouse to Manifest Trustworthy and Green AI), by the Estonian Research Council through the grant PRG1463, by the Estonian Centre of Excellence in Energy Efficiency, ENER (grant TK230) funded by the Estonian Ministry of Education \& Research, and by the European Union and Estonian Research Council via project TEM-TA78.}
\thanks{$^{1}$Komeil Nosrati, Aleksei Tepljakov, and Eduard Petlenkov are with the Department of Computer Systems, Tallinn University of Technology, Estonia (email: {\{komeil.nosrati, aleksei.tepljakov, eduard.petlenkov\}@taltech.ee}).}%
\thanks{$^{2}$Juri Belikov is with the Department of Software Science,
Tallinn University of Technology, Estonia (email: {juri.belikov@taltech.ee}).}}

\begin{document}

\maketitle
\thispagestyle{empty}
\pagestyle{empty}

\begin{abstract}
In model predictive control (MPC), the choice of cost-weighting matrices and designing the Hessian matrix directly affects the trade-off between rapid state regulation and minimizing the control effort. However, traditional MPC in quadratic programming relies on fixed design matrices across the entire horizon, which can lead to suboptimal performance. This study presents a Riccati equation-based method for adjusting the design matrix within the MPC framework, which enhances real-time performance. We employ a penalized least-squares (PLS) approach to derive a quadratic cost function for a discrete-time linear system over a finite prediction horizon. Using the method of weighting and enforcing the equality constraint by introducing a large penalty parameter, we solve the constrained optimization problem and generate control inputs for forward-shifted horizons. This process yields a recursive PLS-based Riccati equation that updates the design matrix as a regularization term in each shift, forming the foundation of the regularized MPC (Re-MPC) algorithm. To accomplish this, we provide a convergence and stability analysis of the developed algorithm. Numerical analysis demonstrates its superiority over traditional methods by allowing Riccati equation-based adjustments.
\end{abstract}
\section{INTRODUCTION}
Model predictive control (MPC) has emerged as a powerful strategy due to its ability to systematically handle constraints and optimize performance over a finite horizon~\cite{ref1}. Its applications span various fields, including process control, robotics, and autonomous systems~\cite{ref2}. A critical component of MPC is the cost function, which incorporates state and input weighting matrices~\cite{ref3}. These matrices dictate the trade-off between rapid system regulation and control effort, thereby directly influencing closed-loop performance~\cite{ref4}. However, traditional MPC frameworks rely on fixed weighting matrices, typically determined through user experience or trial-and-error~\cite{ref5}. This static selection often leads to suboptimal performance, particularly in dynamic or uncertain environments where system conditions and objectives evolve over time~\cite{ref6}. Consequently, systematic methods are increasingly needed to adapt these matrices for robust, efficient control~\cite{ref7}.

Since the weighting matrices determine the Hessian (or design matrix) used in optimization, their static configuration inherently restricts adaptability~\cite{ref23}. Significant research has been conducted to address the challenge of tuning weighting matrices and the resulting design matrix. Traditional methods often involve manual tuning or heuristic approaches, which are time-consuming and lack theoretical guarantees~\cite{ref8}. More advanced techniques include optimization-based methods, such as Bayesian optimization~\cite{ref9}, multi-objective Bayesian optimization~\cite{ref10}, and reinforcement learning~\cite{ref11}, which aim to automate the tuning process. Additionally, adaptive MPC frameworks have been proposed, where weighting matrices are adjusted online based on system performance or environmental changes~\cite{ref6}. While these methods offer improvements, they often suffer from computational inefficiency and scalability issues, particularly in high-dimensional control problems~\cite{ref13}:
\begin{itemize}
    \item Although analytical results exist for simple systems~\cite{ref14}, real-world applications often involve high-order or marginally stable dynamics, making explicit tuning formulas impractical.
    \item Multi-variable systems introduce high-dimensional parameter spaces with interactions, limiting the effectiveness of single-scalar tuning or heuristic methods~\cite{ref15}.
    \item User performance requirements evolve over time, requiring a systematic approach that accommodates multiple objectives instead of relying on trial-and-error~\cite{ref17}.
\end{itemize}
In response to the challenges of automatic MPC cost function regularization---eliminating the need for time-consuming trial-and-error and adapting to changing environmental conditions beyond the capability of fixed-cost controllers~\cite{ref18}---this work introduces a novel regularized MPC (Re-MPC) framework for linear time-invariant (LTI) systems to dynamically update the design matrix through a penalized least squares (PLS)-based Riccati equation. By employing this recursive equation, the framework iteratively adjusts the design matrix, effectively incorporating a regularization term at each time step. The key contributions are as follows:
\begin{itemize}
    \item A least squares problem with equality constraints (LSE) was developed for Re-MPC design using a weighting method and a penalty parameter to enforce constraints.
  \item The existence condition was examined, and a PLS-based Riccati equation was derived to update the design matrix and generate control inputs for shifted horizons.
  \item The unique positive definite (PD) solution of the Riccati equation and stability of the closed-loop system were proven under controllability and detectability assumptions using strong induction and proof by contradiction.
\end{itemize}
Compared to existing MPC techniques~\cite{ref23}, the developed Re-MPC algorithm offers superior state regulation and control efficiency by continuously updating the design matrix, an effective solution for control problems with adaptable closed-loop dynamics. Within the algorithm, the recursive Riccati equation updates this matrix as a regularization term in each iteration, ensuring dynamic adaptability and enhanced stability. While a penalty parameter is used to enforce the equality constraint, a key advantage of the Re-MPC is the elimination of the need for this auxiliary parameter and its absence in the final result when deriving the exact solution. This advantage, along with the matrix-based approach to feedback gain and Riccati equation derivation, makes the Re-MPC useful for online applications and a basis for future robust MPC designs for uncertain systems.

\textit{Notations}: $\mathbb{N}$ and $\mathbb{C}$ represent sets of natural and complex numbers. $\mathbb{R}^n$ and $\mathbb{R}^{n \times m}$ are the real $n$-dimensional vector and $n \times m$ matrix, respectively. $A^\mathsf{T}$ is the transpose of $A$, and $B = B^\mathsf{T} \succ 0$ is a symmetric PD matrix. $C \succeq 0$ and $D \preceq 0$ denote positive semidefinite (PSD) and negative semidefinite (NSD) matrices, respectively. $I_n$ is the identity matrix of order $n$, $\mathbf{1}_n$ is a column vector of ones of length $n$, and $S_{n}$ is the lower shift matrix with $s_{ij} = \delta_{i,j+1}$, where $\delta$ is the Kronecker delta. $t_f \in \mathbb{N}$ and $l \in \mathbb{N}$, $0 \leq l \leq t_f$, are the time and prediction horizons, respectively. $\col\{A_1,A_2\} = \col\{A_j\}_{j=1}^2$ denotes a column vector, $x^*$ is the optimal value of $x$, $\left\|x\right\|_A^2 = x^\mathsf{T} A x$ is the weighted norm, and $\mathcal{N}(A)$ is the null space of $A$. Moreover, $\oplus$ and $\otimes$ denote the direct sum and Kronecker product, and $\dagger$ denotes the Moore-Penrose pseudoinverse.

\section{Problem Formulation}
Consider a discrete LTI system over a finite time horizon of length $t_f$, as described by
\begin{equation}\label{eq1}
\begin{aligned}
x_{k+1} &= Ax_{k}+Bu_{k},  \quad   x_0  \in \mathbb{R}^{n},  \quad 0 \leq k \leq t_f,
\end{aligned}
\end{equation}
where $x  \in \mathbb{R}^{n}$ is the state vector, $u  \in \mathbb{R}^{m}$ is the control input, and $A \in \mathbb{R}^{n \times n}$ and $B \in \mathbb{R}^{n \times m}$ are constant matrices. We aim to design a control input $u_k$ for system \eqref{eq1}, using the proposed Re-MPC method, to stabilize the state $x_k$ at the origin. To achieve this, we first introduce a set of assumptions regarding the system dynamics and terminal constraints, which form the foundation for the Re-MPC design.

\begin{assumption}\label{ass1}
The state and control input are subject to linear inequality constraints represented by
\begin{equation}\label{eq2}
\begin{aligned}
\mathbb{X}=\{x \in \mathbb{R}^n | F_x x \leq g_x\}, \quad   \mathbb{U}=\{u \in \mathbb{R}^m | F_u u \leq g_u\},
\end{aligned}
\end{equation}
where $F_x \in \mathbb{R}^{2n \times n}$ and $g_x \in \mathbb{R}^{2n}$ are constant matrix and vector, respectively, representing the state constraints, and $F_u \in \mathbb{R}^{2m \times m}$ and $g_u \in \mathbb{R}^{2m}$ are constant matrix and vector, respectively, specifying the input constraints.
\end{assumption}
\begin{assumption}[\cite{ref19}]\label{ass2} The LTI system~\eqref{eq1} is controllable, i.e., for all $z \in \mathbb{C}$, we have $\rank(\begin{bmatrix}zI-A & B\end{bmatrix})=n$.
\end{assumption}

In the context of a receding horizon or MPC approach over the prediction horizon $l$, we set the following cost function:
\begin{equation}\label{eq3}
\begin{aligned}
J_k = \|x_{k+l|k}\|_{P_{k+l}}^2+ \sum_{i=k}^{k+l-1} \|x_{i|k}\|_{Q}^2+\|u_{i|k}\|_{R}^2,
\end{aligned}
\end{equation}
where $P_{k+l} \succ 0$, $Q \succeq 0$, and $R \succ 0$ represent PD weighting matrices. Based on this definition, we introduce another assumption as follows:

\begin{assumption}[\cite{ref19}]\label{ass3} The pair $(A,Q)$ is detectable, i.e., for all $z \geq 1 \in \mathbb{C}$, we have $\rank(\col\{zI-A,Q\})=n$.
\end{assumption}

We define the predicted state and control sequences as
\begin{equation}\label{eq4}
\begin{aligned}
X_{k}=\col\left\{x_{k+j|k}\right\}_{j=1}^l, \quad U_{k}=\col\left\{u_{k+j|k}\right\}_{j=0}^{l-1},
\end{aligned}
\end{equation}
where the initial state $x_{k|k}$ corresponds to the current state at the $k^{\text{th}}$ horizon. Based on this setup, we aim to achieve the following two objectives:
\begin{itemize}
    \item \textit{Objective I}: For the discrete LTI system~\eqref{eq1} with $x_{k|k}=x_k$, the goal is to compute the optimal sequences of states $X_{k}^{*}$ and control inputs $U_{k}^{*}$ by solving the following optimization problem:
    \begin{equation}\label{eq5}
    \begin{aligned}
    &\argmin_{U_{k},X_{k}}J_k,\\
    & \text{s.t.} \quad x_{i+1|k}= Ax_{i|k}+Bu_{i|k}, \,\, \, x_{i|k} \in \mathbb{X},\\
    &X_{k} \in \mathbb{X}^{l}, \, \, U_{k} \in \mathbb{U}^{l}, \, \, k \le i \le k+l-1,
    \end{aligned}
    \end{equation}
    which leads to the formulation of the Re-MPC algorithm.
    \item \textit{Objective II}: Identify the conditions under which the obtained sequences $X_{k}^{*}$ and $U_{k}^{*}$ guarantee the convergence and stability of the algorithm.   
\end{itemize}
To accomplish these objectives, we employ a PLS-based framework, which effectively addresses the constrained optimization problem~\eqref{eq5}. By re-examining the standard algebraic Riccati equation and utilizing the obtained feedback gains, we establish the necessary conditions for the stability and convergence of the developed Re-MPC algorithm.

\section{Re-MPC Algorithm}
To address the first objective, we introduce a sequence of key lemmas that establish fundamental connections essential to solving our problem. We start by considering a general quadratic optimization problem
\begin{equation}\label{eq6}
\argmin_{\eta} \{ (G\eta - h)^\mathsf{T} W (G\eta - h) \}
\end{equation}
where $\eta \in \mathbb{R}^{m}$ is an unknown vector, $h \in \mathbb{R}^{n}$ and $G \in \mathbb{R}^{n \times m}$ are known vector and matrix, respectively, and $W \in \mathbb{R}^{n \times n}$ is a known PD weighting matrix.

\begin{lemma}[\cite{ref20}]\label{lem1}
Let the matrix $G$ have full column rank, i.e., $\rank (G) = m$, and define $\Theta_1 = G^\mathsf{T} W G$ and 
\begin{equation}\label{eq7}
\Theta_2 = \begin{bmatrix} W^{-1} & G \\ G^\mathsf{T} & 0 \end{bmatrix}.
\end{equation}
Then, $\Theta_1$ and $\Theta_2$ are invertible matrices. This guarantees that the optimization problem~\eqref{eq6} has a unique optimal solution, given by\footnote{In the case that $W \succeq 0$, the solution is $\hat{\eta} = (G^\mathsf{T} W G)^{\dagger} G^\mathsf{T} W h$. Thus, for a PSD matrix $Q \succeq 0$, the pseudoinverse is used, but its definiteness does not affect the derivation of the algorithm, as will be evident from the results. For clarity, we assume $Q \succ 0$ to use standard inverse notation.}
\begin{equation}\label{eq8}
\begin{aligned}
   \hat{\eta}=\begin{bmatrix}0 \\ I\end{bmatrix}^\mathsf{T}\begin{bmatrix} W^{-1} & G \\ G^\mathsf{T} & 0 \end{bmatrix}^{-1}\begin{bmatrix} h \\ 0\end{bmatrix}= (G^\mathsf{T} W G)^{-1} G^\mathsf{T} W h.
    \end{aligned}
\end{equation}
\end{lemma}

\begin{lemma}[\cite{ref21}]\label{lem2}
Consider a constrained version of the problem~\eqref{eq6}, where we introduce the equality constraint $F \eta= \varphi$, with $F \in \mathbb{R}^{k \times m}$ and $\varphi \in \mathbb{R}^{k}$. According to~\cite{ref22}, a solution to this problem is unique if and only if $\rank(F) = k$ and $\rank(\col\{G,F\}) = m$, or equivalently, $\mathcal{N}(G) \cap \mathcal{N}(F) = \{0\}$. Using the weighting method for solving LSE, we enforce the constraint by introducing a large penalty parameter $\mu$. This reformulates the problem as an unconstrained least squares problem and leads to the following PLS-based formulation:  
\begin{equation}\label{eq9}
\argmin_{\eta} \{ (\Bar{G}\eta - \Bar{h})^\mathsf{T} \Bar{W} (\Bar{G}\eta - \Bar{h}) \},
\end{equation}
where $\Bar{G} = \col\{G, F\}$, $\Bar{h} = \col\{h, \varphi\}$, and $\Bar{W} = (W \oplus \mu I)$. From Lemma~\ref{lem1}, the unique solution to~\eqref{eq9} is:
\begin{equation}\label{eq10}
\begin{aligned}
\hat{\eta}_\mu = \begin{bmatrix}0 & I\end{bmatrix} 
\begin{bmatrix} \Bar{W}^{-1} & \Bar{G} \\ 
\Bar{G}^\mathsf{T} & 0 \end{bmatrix}^{-1} 
\begin{bmatrix} \Bar{h} \\ 0\end{bmatrix}.
\end{aligned}
\end{equation}
For large values of $\mu$, we obtain $\lim_{\mu \to \infty} \hat{\eta}_\mu = \hat{\eta}_{\text{LSE}}$, where $\hat{\eta}_{\text{LSE}}$ is the exact optimal solution of~\eqref{eq6} subject to the equality constraint $F \eta= \varphi$, given by  
\begin{equation}\label{eq11}
\begin{aligned}
   \hat{\eta}=\begin{bmatrix}0 \\ 0 \\ I\end{bmatrix}^\mathsf{T}\begin{bmatrix} W^{-1} & 0 & G \\ 0 & 0 & F \\ G^\mathsf{T}  & F^\mathsf{T} & 0 \end{bmatrix}^{-1}\begin{bmatrix} h \\ \varphi \\ 0\end{bmatrix}.
    \end{aligned}
\end{equation}
\end{lemma}

According to Lemma~\ref{lem2}, the penalty parameter transforms constrained problems into unconstrained ones by penalizing deviations from the equality constraints within the objective function. Building on these foundational results, we now present a recursive approach to solving~\eqref{eq5}, which leads to the derivation of the Re-MPC formulation.

\begin{theorem}\label{thm1}
Under inactive inequality constraints, the optimal recursive solution to the problem~\eqref{eq5} is given by
\begin{equation}\label{eq12}
\Bar{U}_k^* = K_k x_{k|k}, \quad 0 \leq k \leq t_f,
\end{equation}
where $\Bar{U}_k^*$ stacks optimal decision variables as $\Bar{U}_{k}^{*}=\col\{X_{k}^{*},U_{k}^{*}\}$, and $K_k=\col\{K_{X_k},K_{U_k}\}$ is a gain matrix:
\begin{equation}\label{eq13}
K_k = \begin{bmatrix}0 & I \end{bmatrix}
\begin{bmatrix}
\mathscr{H}^{-1} & \mathscr{B} \\ 
\mathscr{B}^\mathsf{T} & 0
\end{bmatrix}^{-1}
\begin{bmatrix}\mathscr{A} \\ 0 \end{bmatrix}.
\end{equation}
Here, the augmented matrices $\mathscr{A}=\col\{0,\mathscr{A}_2\}$ and $\mathscr{B}=\col\{I,\mathscr{B}_2\}$ encode the system dynamics~\eqref{eq1} as
\begin{equation}\label{eq14}
\begin{aligned}
\mathscr{A}_2&=\begin{bmatrix}-I \\ \Bar{A} \end{bmatrix}, \quad \quad \mathscr{B}_2=\begin{bmatrix} 0 & 0 \\ \Bar{B}_1 & -\Bar{B}_2 \end{bmatrix},
\end{aligned}
\end{equation}
with $\Bar{A}=\col\{A,0,\ldots,0\} \in \mathbb{R}^{\Bar{n} \times n}$, $\Bar{B}_2=\diag\{B,\ldots,B\}\in \mathbb{R}^{\Bar{n} \times \Bar{m}}$, and $\Bar{B}_1=I-\Tilde{A} \in \mathbb{R}^{\Bar{n} \times \Bar{n}}$, where $\Tilde{A}=\diag\{A,\ldots,A\}S_{\Bar{n}} \in \mathbb{R}^{\Bar{n} \times \Bar{n}}$ in which $\Bar{n}=l \, n$, and $\Bar{m}=l \, m$. Also, the augmented matrix $\mathscr{H} = (H_1 \oplus H_2)$ encodes the weighting matrices as
\begin{equation}\label{eq15}
\begin{aligned}
H_1&=\begin{bmatrix} \Bar{Q}_{k+l} & 0 \\ 0 & \Bar{R}\end{bmatrix}, \quad  H_2=\begin{bmatrix}Q & 0 \\0 & \mu I \end{bmatrix},
\end{aligned}
\end{equation}
where $\Bar{Q}_{k+l}=\diag\{Q,\ldots,Q,P_{k+l}\} \in \mathbb{R}^{\Bar{n} \times \Bar{n}}$ and $\Bar{R}=\diag\{R,\ldots,R\} \in \mathbb{R}^{\Bar{M} \times \Bar{M}}$. Furthermore, the minimum cost associated with the optimal solution is $J^{*}=x_{k|k}^\mathsf{T}P_{k+l-1}x_{k|k}$, where $P_{k+l-1}=K_{k}^\mathsf{T}H_{1} K_{k}+Q+\mu\mathscr{R}^\mathsf{T}\mathscr{R}$ with $\mathscr{R}=\Bar{B}_1K_{X_k}-\Bar{B}_2K_{U_k}-\Bar{A}$ representing the residual term.
\end{theorem}

\begin{proof}
When all inequality constraints are inactive, we can rewrite the problem~\eqref{eq5} as
\begin{equation}\label{eq16}
\begin{aligned}
\argmin_{U_{k},X_{k}}\quad &X_k^\mathsf{T} \Bar{Q}_{k+l} X_k + U_k^\mathsf{T} \Bar{R} U_k + x_{k|k}^\mathsf{T} Q x_{k|k},\\
\text{s.t.} \quad &\Bar{B}_1X_k = \Bar{A} x_{k|k} + \Bar{B}_2 U_k,
\end{aligned}
\end{equation}    
To derive the exact optimal solution~\eqref{eq12}, we rewrite the constrained optimization problem~\eqref{eq16} in an alternative form with a large penalty parameter $\mu$, given by
\begin{equation}\label{eq17}
\begin{aligned}
\Bar{U}_k^{*} &= \underset{\Bar{U}_k}{\argmin} \left\{ \|\Bar{U}_k\|_{H_1}^2 + \|\mathscr{B}_2 \Bar{U}_k - \mathscr{A}_2 x_{k|k} \|_{H_2}^2 \right\} \\  
&= \underset{\Bar{U}_k}{\argmin} \{ (\mathscr{B} \Bar{U}_k - \mathscr{A} x_{k|k})^\mathsf{T} \mathscr{H} (\mathscr{B} \Bar{U}_k - \mathscr{A} x_{k|k}) \}.
\end{aligned}
\end{equation}  
By applying Lemma~\ref{lem2}, we obtain the optimal solution as~\eqref{eq12}. Moreover, substituting this solution into~\eqref{eq17} yields
\begin{equation}\label{eq18}
\begin{aligned}
J^{*} &= (\mathscr{B} K_k x_{k|k} - \mathscr{A} x_{k|k})^\mathsf{T} \mathscr{H} (\mathscr{B} K_k x_{k|k} - \mathscr{A} x_{k|k})\\  
&= (\Bar{K}_k x_{k|k})^\mathsf{T} \mathscr{H} (\Bar{K}_k x_{k|k})\\  
&= x_{k|k}^\mathsf{T} \big( K_k^\mathsf{T} H_1 K_k + Q + \mu \mathscr{R}^\mathsf{T} \mathscr{R} \big) x_{k|k},  
\end{aligned}
\end{equation}  
where $\Bar{K}_k = \col\{K_k, I, \mathscr{R}\}$. The proof follows for
\begin{equation}\label{eq19}
\begin{aligned}
 P_{k+l-1}=K_k^\mathsf{T} H_1 K_k + Q + \mu \mathscr{R}^\mathsf{T} \mathscr{R}.
\end{aligned}
\end{equation} 
\end{proof}

\begin{remark}\label{rem1}
The formulation of Re-MPC is inspired by the structural characteristics of classical quadratic-cost designs, combined with the penalty function approach (see~\cite{ref22}). This approach enables the extension of recursive Riccati equation-based control to MPC. As $\mu \to \infty$, equation~\eqref{eq17} converges to the constrained problem~\eqref{eq16} (see Lemma~\ref{lem2}). Within this penalized framework, the last term in~\eqref{eq19} vanishes as $\mathscr{R} \to 0$, improving controller accuracy and yielding an exact solution to the constrained problem that satisfies the system in~\eqref{eq1}. This leads to the following recursive Riccati equation
\begin{equation}\label{eq20}
\begin{aligned}
P_{k+l-1} = K_{X_k}^\mathsf{T} \Bar{Q}_{k+l} K_{X_k} + K_{U_k}^\mathsf{T} \Bar{R} K_{U_k} + Q.
\end{aligned}
\end{equation}
\end{remark}

\begin{remark}\label{rem1new}
When inequality constraints are active, problem~\eqref{eq16} includes the additional condition $\Bar{F} \Bar{U}_k \leq \Bar{g}$, where $\Bar{F} = (F_{X_k} \oplus F_{U_k})$ and $\Bar{g} = (g_{X_k} \oplus g_{U_k})$, with $F_{X_k} = (I_l \otimes F_x)$, $F_{U_k} = (I_l \otimes F_u)$, $g_{X_k} = (\mathbf{1}l \otimes g_x)$, and $g_{U_k} = (\mathbf{1}_l \otimes g_u)$. This problem can be solved numerically using {\tt fmincon} in MATLAB\footnote{Specifically, the problem can be solved as $\arg \min_{\Bar{U}_{k}} f(\Bar{U}_k)=\Bar{U}_k^{\mathsf{T}}H_1 \Bar{U}_k + cte$ with the equality constraint $\Bar{B}_1 X_k = \Bar{A} x_{k|k} + \Bar{B}_2 U_k$ rewritten as $\begin{bmatrix}\Bar{B}_1 & -\Bar{B}_2\end{bmatrix} \Bar{U}_k = \Bar{A} x_{k|k}$, which corresponds to the $A_{eq} x = b_{eq}$ formulation in {\tt fmincon}, while the inequality constraint $\Bar{F} \Bar{U}_k \leq \Bar{g}$ maps directly to $A x \leq b$.}, and the cost-to-go matrix $P_{k+l-1}$ in~\eqref{eq20} contributes to the design of the Re-MPC by updating the terminal cost matrix $\Bar{Q}_k$ and the design matrix $H_1$ via~\eqref{eq15}, effectively serving as a regularization term in each shift.
\end{remark}

Building on this insight and Theorem~\ref{thm1}, we propose the Re-MPC design for the discrete LTI system in~\eqref{eq1}, summarized in Algorithm\footnote{The operator $[\cdot]_1$ outputs only the first entry of the vector $[\cdot]$.}~\ref{alg1}. In contrast to traditional MPC~\cite{ref23}, Re-MPC adjusts the design matrix $H_{1}$ at each iteration based on the solution of equation~\eqref{eq19}, thereby enhancing the optimization process in the predictive control. In our algorithm, as the parameter $\mu$ tends to infinity, it effectively vanishes within the algorithm. This procedure ensures that Re-MPC converges and remains stable at each time step, regardless of the value of $\mu$ (see the next section for proof).

\begin{algorithm}[t]
\vspace*{0.4mm}
\caption{Regularized MPC Algorithm.}\label{alg1}
\begin{algorithmic}[1]
\STATE Input: $A$, $B$, $l$, $P_{k+l}$, $Q$, $R$, $F_x$, $g_x$, $F_u$, $g_u$, $\mu$
\STATE Initialization: $x_{0|0}=x_0$
\STATE Construct $\Bar{A}$, $\Bar{B}_1$, $\Bar{B}_2$, $\Bar{Q}_{k+l}$, $\Bar{R}$, $\Bar{F}$, $\Bar{g}$
\STATE Construct $X_k$, $U_k$, and $\Bar{U}_{k}=\col\{X_{k},U_{k}\}$
\STATE Construct matrix $H_1$ and $F_{eq}=\begin{bmatrix} \Bar{B}_1 & -\Bar{B}_2 \end{bmatrix}$

\FOR {$k=0$ to $t_f-1$}
\STATE Compute the MPC gain $K_k$ by~\eqref{eq13}.
\STATE Define the function $f(\Bar{U}_k)=\Bar{U}_k^{\mathsf{T}}H_1 \Bar{U}_k$.
\STATE Compute $\Bar{U}_{k}^{*}$ using~\eqref{eq12} (in-active inequalities).
\STATE Compute $\Bar{U}_{k}^{*}$ using {\tt fmincon} (active inequalities).
\STATE Set $u_k=u_{k|k}=[U_k^*]_1$ and $x_{k+1}=x_{k+1|k}=[X_k^*]_1$.
\STATE Compute $P_{k+l-1}$ using~\eqref{eq19} backward in time.
\STATE Update $\Bar{Q}_k$ and $H_1$ using~\eqref{eq15}.
\ENDFOR
\end{algorithmic} 
\end{algorithm}

\section{Convergence and Stability of Re-MPC}
To meet Objective II, we derive sufficient conditions for stability and convergence of the Re-MPC algorithm for a large value of $\mu$. To do so, we analyze the Riccati equation~\eqref{eq20} along with the obtained feedback gains~\eqref{eq13}. Using Lemma~\ref{lem1} and some algebraic simplifications, we obtain
\begin{equation}\label{eq21}
\begin{aligned}
\begin{bmatrix}
    X_{k}^{*} \\ U_{k}^{*}
\end{bmatrix} =\begin{bmatrix} \Bar{Q}_{k+l} & 0 \\ 0 & \Bar{R}\end{bmatrix}^{-1}\begin{bmatrix} \Bar{B}_1 & -\Bar{B}_2 \end{bmatrix}^\mathsf{T}O_{k+l}^{-1}\Bar{A}x_{k|k},
\end{aligned}
\end{equation}
where $O_{k+l} = \Bar{B}_1\Bar{Q}_{k+l}^{-1}\Bar{B}_1^\mathsf{T}+\Bar{B}_2\Bar{R}^{-1}\Bar{B}_2^\mathsf{T}$.

\begin{remark}\label{rem3}
To ensure the feasibility of the Re-MPC algorithm, it is crucial that the matrix $O_{k+l}$ is invertible. Given that $\Bar{Q}_{k+l} \succ 0$ and $\Bar{R} \succ 0$, this requirement is satisfied if $\text{rank} (\begin{bmatrix} \Bar{B}_1 & -\Bar{B}_2 \end{bmatrix}) = \Bar{n}$. Given the structure of $\Bar{B}_1$ as a lower triangular matrix with an identity matrix on its diagonal, it is always full-rank. As a result, the rank condition is always satisfied as long as $\Bar{B}_2$ does not introduce rank deficiency.
\end{remark}

\begin{remark}\label{rem4}
By applying equation~\eqref{eq21} and defining the augmented matrices in Theorem~\ref{thm1} as index matrices, i.e., $\Bar{A}_l := \Bar{A}$, $\Bar{B}_{i,l} := \Bar{B}_i$ for $i = 1, 2$, and $\Bar{R}_l := \Bar{R}$, within a prediction horizon of size $l$, and after some algebraic manipulations, the Riccati equation~\eqref{eq20} can be rewritten as
\begin{equation}\label{eq22}
\begin{aligned}
P_{k+l-1}&=\Bar{A}_l^\mathsf{T}(\Bar{B}_{1,l}\Bar{Q}_{k+l}^{-1}\Bar{B}_{1,l}^\mathsf{T}+\Bar{B}_{2,l}\Bar{R}_l^{-1}\Bar{B}_{2,l}^\mathsf{T})^{-1}\Bar{A}_l+Q,
\end{aligned}
\end{equation}
or alternatively, in a steady-state scenario,
\begin{equation}\label{eq23}
P_l=\Bar{A}_l^\mathsf{T}(\Bar{B}_{1,l}\Bar{Q}_l^{-1}\Bar{B}_{1,l}^\mathsf{T}+\Bar{B}_{2,l}\Bar{R}_{l}^{-1}\Bar{B}_{2,l}^\mathsf{T})^{-1}\Bar{A}_l+Q.
\end{equation}
\end{remark}

\begin{theorem}\label{thm2}
Under Assumptions~\ref{ass2} and \ref{ass3}, the Riccati equation~\eqref{eq23} admits a unique solution $P_l \succ 0$ for the entire process as the prediction horizon shifts within the total time horizon $t_f$.
\end{theorem}

\begin{proof}
To prove this theorem, we use complete (strong) induction, starting with the base cases for $l = 1$ and $l = 2$, assuming the validity of the statement for all values up to $l - 1$, and then proving the result for $l$.

\textit{Base step ($l=1$ and $l=2$):} For $l = 1$, the Riccati equation in~\eqref{eq23} can be rewritten as
\begin{equation}\label{eq24}
\begin{aligned}
P_1&=A_1^\mathsf{T}(B_{1,1}P_1^{-1}B_{1,1}+B_{2,1}R_1^{-1}B_{2,1}^\mathsf{T})^{-1}A_1+Q\\
&=A^\mathsf{T}(P_1^{-1}+BR^{-1}B^\mathsf{T})^{-1}A+Q,
\end{aligned}
\end{equation}
which is a reformulation of the standard Riccati equation
\begin{equation}\label{eq25}
P_1=A^\mathsf{T}P_1A - A^\mathsf{T}P_1B (R+B^\mathsf{T}P_1B)^{-1}B^\mathsf{T}P_1A +Q,
\end{equation}
obtained by applying the Woodbury matrix identity. Therefore, under Assumptions~\ref{ass2} and \ref{ass3}, and given the positive definiteness of the weighting matrices, the Riccati equation for the Re-MPC at this step has a unique PD solution $P_1 \succ 0$. For $l = 2$, the Riccati equation in~\eqref{eq23} can be written as 
\begin{equation}\label{eq26}
\begin{aligned}
P_2&=\begin{bmatrix} A \\ 0 \end{bmatrix}^\mathsf{T}\begin{bmatrix} \mathscr{L} & -Q^{-1}A^\mathsf{T} \\
-AQ^{-1} & \mathscr{P}+AQ^{-1}A^\mathsf{T}\end{bmatrix}^{-1}\begin{bmatrix} A \\ 0 \end{bmatrix}+Q\\
&=A^\mathsf{T}((Q+A^\mathsf{T}\mathscr{P}^{-1}A)^{-1}+BR^{-1}B^\mathsf{T})^{-1}A+Q,
\end{aligned}
\end{equation}
where $\mathscr{L} = Q^{-1}+BR^{-1}B^\mathsf{T}$ and $\mathscr{P} = P_2^{-1}+BR^{-1}B^\mathsf{T}$. According to the base step 1, the term $Q+A^\mathsf{T}\mathscr{P}^{-1}A$ has a unique PD solution, which is given by $P_2$. Therefore, we conclude that 
\begin{equation}\label{eq27}
\begin{aligned}
P_2=A^\mathsf{T}(P_2^{-1}+BR^{-1}B^\mathsf{T})^{-1}A+Q,
\end{aligned}
\end{equation}
and as established for $l = 1$, this equation admits a unique PD solution, $P_2 \succ 0$.

\textit{Assumption step ($l-1$):} The Riccati equation~\eqref{eq23} becomes
\begin{equation}\label{eq28}
\begin{aligned}
P_{l-1}=\Bar{A}_{l-1}^\mathsf{T}(\Bar{B}&_{1,l-1}\Bar{Q}_{l-1}^{-1}\Bar{B}_{1,l-1}^\mathsf{T}\\
                &+\Bar{B}_{2,l-1}\Bar{R}_{l-1}^{-1}\Bar{B}_{2,l-1}^\mathsf{T})^{-1} \Bar{A}_{l-1}+Q.
\end{aligned}
\end{equation}
Let us assume the statement holds for this equality, with a unique PD solution $P_{l-1} \succ 0$.

\textit{Induction step ($l$):} One can rewrite the equation~\eqref{eq23} as 
\begin{equation}\label{eq29}
\begin{aligned}
P_{l}&=\begin{bmatrix} A \\ 0 \end{bmatrix}^\mathsf{T} \begin{bmatrix} \Bar{\mathscr{L}} & Q^{-1}\Bar{A}_{l-1}^\mathsf{T} \\
\Bar{A}_{l-1}Q^{-1} & \Bar{\mathscr{P}}+\Bar{A}_{l-1}Q^{-1}\Bar{A}_{l-1}^\mathsf{T} \end{bmatrix}^{-1} \begin{bmatrix} A \\ 0 \end{bmatrix}+Q\\
&=A^\mathsf{T}((Q+\Bar{A}_{l-1}^\mathsf{T}\Bar{\mathscr{P}}^{-1}\Bar{A}_{l-1})^{-1}+BR^{-1}B^\mathsf{T})^{-1}A+Q,
\end{aligned}
\end{equation}
where $\Bar{\mathscr{L}} = Q+BR^{-1}B^\mathsf{T}$ and $\Bar{\mathscr{P}} = \Bar{B}_{1,l-1}\mathscr{Q}_l^{-1}\Bar{B}_{1,l-1}^\mathsf{T}+\Bar{B}_{2,l-1}\Bar{R}_{l-1}^{-1}\Bar{B}_{2,l-1}^\mathsf{T}$, with $\mathscr{Q}_l=\diag\{Q,\cdots,Q,P_{l}\} \in \mathbb{R}^{\overline{l-1} \times \overline{l-1}}$. According to step $l-1$, the term $Q+\Bar{A}_{l-1}^\mathsf{T}\Bar{\mathscr{P}}^{-1}\Bar{A}_{l-1}$ has a unique PD solution $P_l$. Therefore, we can rewrite~\eqref{eq29} as
\begin{equation}\label{eq30}
\begin{aligned}
P_l=A^\mathsf{T}(P_l^{-1}+BR^{-1}B^\mathsf{T})^{-1}A+Q.
\end{aligned}
\end{equation}
According to the proof in step 1, this equation, which is identical to~\eqref{eq25}, has a unique PD solution since $P_l \succ 0$.
\end{proof}

To prove the stability of the closed-loop system under the developed Re-MPC algorithm, we establish the stability of the system in~\eqref{eq21}. Using strong induction, we show that Assumptions~\ref{ass2} and \ref{ass3} ensure system stability.

\begin{theorem}\label{thm3}
Given the controllability and detectability conditions, the developed Re-MPC Algorithm~\ref{alg1} is stable.
\end{theorem}

\begin{proof}
From Remark~\ref{rem1},  as $\mu \to \infty$, the residual term $\mathscr{R}$ approaches zero, implying that
\begin{equation}\label{eq31}
\begin{aligned}
\Bar{B}_1K_{X_k}=\Bar{A}+\Bar{B}_2K_{U_k}.
\end{aligned}
\end{equation}

\textit{Base step ($l=1$ and $l=2$):} For $l = 1$, the closed-loop system in~\eqref{eq21} can be stated as
\begin{equation}\label{eq32}
\begin{aligned}
x_{k+1|k}^{*}&=K_{x_{k+1|k}}x_{k|k}\\
&=P_{k+1}^{-1}(P_{k+1}^{-1}+BR^{-1}B^\mathsf{T})^{-1}Ax_{k|k}.
\end{aligned}
\end{equation}
Additionally, the recursive Riccati equation~\eqref{eq20} reduces to
\begin{equation}\label{eq33}
\begin{aligned}
P_{1}=K_{x_1}^\mathsf{T}P_{1}K_{x_1}+K_{u_0}^\mathsf{T}R K_{u_0}+Q,
\end{aligned}
\end{equation}
where $P_1 \succ 0$, as established in Theorem~\ref{thm2}, step 1. Suppose the closed-loop system~\eqref{eq32} is unstable, i.e., $\exists z \in \mathbb{C}, \quad \exists \rho^{\mathsf{H}} \in \mathbb{C}^n$ such that $|z| \geq 1$ and $\rho K_{x_1} = z \rho$. By pre-multiplying equation~\eqref{eq33} by $\rho$ and post-multiplying by $\rho^\mathsf{H}$, we obtain $(1-|z|^2)\rho P_1 \rho^\mathsf{H}=\rho(K_{u_0}^\mathsf{T}R K_{u_0}+Q)\rho^\mathsf{H}$, where the left-hand side is NSD and the right-hand side is PSD. Thus, $\rho (K_{u_0}^\mathsf{T} R K_{u_0} + Q) \rho^\mathsf{H} = 0$. This, along with $\rho(K_{x_1} - A) \rho^\mathsf{H} = \rho B K_{u_0} \rho^\mathsf{H}$, derived from~\eqref{eq31} for $l = 1$ and pre- and post-multiplied by $\rho$ and $\rho^\mathsf{H}$, implies
\begin{equation}\label{eq34}
\begin{aligned}
\rho\begin{bmatrix} I \\ I \\ K_{u_0} \end{bmatrix}^\mathsf{T}
\begin{bmatrix}
 z I-A & -B \\ Q & 0 \\
    0 & R
\end{bmatrix}\begin{bmatrix} I \\ K_{u_0} \end{bmatrix}\rho^{\mathsf{H}}=0.
\end{aligned}
\end{equation}
Since the left and right matrices are nonzero, the equality in~\eqref{eq36} holds only if the middle matrix has a nontrivial null space, contradicting Assumptions~\ref{ass2} and~\ref{ass3}. Therefore, $K_{x_1}$ and the resulting closed-loop system for $l = 1$ are stable.

For $l = 2$, we define the augmented vectors in~\eqref{eq4} as index vectors, i.e., $X_l := X_k$ and $U_l := U_k$, within a prediction horizon of size $l$. The closed-loop system in~\eqref{eq21} is then given by $X_{2}^{*}=K_{X_2}x_{k|k}$, where $X_{2}^{*}=\col\{x_{k+1|k}^{*},x_{k+2|k}^{*}\}$ and $K_{X_2}=\col\{K_{x_{k+1|k}},K_{x_{k+2|k}}\}$. Utilizing the Woodbury matrix identity and block matrix inversion, we obtain
\begin{equation}\label{eq35}
\begin{aligned}
\begin{bmatrix}x_{k+1|k}^{*} \\ x_{k+2|k}^{*}\end{bmatrix}=\begin{bmatrix} P_{k+2}^{-1}(P_{k+2}^{-1}+BR^{-1}B^\mathsf{T})^{-1}A \\ (P_{k+2}^{-1}(P_{k+2}^{-1}+BR^{-1}B^\mathsf{T})^{-1}A)^2\end{bmatrix}x_{k|k}.
\end{aligned}
\end{equation}
The first equation in~\eqref{eq35} mirrors the closed-loop system for $l=1$, but with $P_{k+2}$ instead of $P_{k+1}$. By Theorem~\ref{thm2}, step 1 and $l=2$, the Riccati equation has a steady-state PD solution $P_2 \succ 0$. By analogy to step 1, this equation is stable. Moreover, the gain in the second equation follows the same form as the first but squared. Since its eigenvalues remain within the unit circle, stability is preserved. Thus, the gain matrix $K_{X_2}$ and the closed-loop system remain stable.

\textit{Assumption step ($l-1$):} At this step, the closed-loop system in~\eqref{eq21} is $X_{l-1}^{*}=K_{X_{l-1}}x_{k|k}$, where $X_{l-1}^{*}=\col\{x_{k+j|k}^{*}\}_{j=1}^{l-1}$ and $K_{X_{l-1}}=\col\{K_{x_{k+j|k}}\}_{j=1}^{l-1}$. Assuming that
\begin{equation}\label{eq36}
\begin{aligned}
K_{X_{l-1}}=\Bar{Q}_{l-1}\Bar{B}_{1,l-1}^\mathsf{T}(&\Bar{B}_{1,l-1}\Bar{Q}_{l-1}^{-1}\Bar{B}_{1,l-1}^\mathsf{T}\\
                &+\Bar{B}_{2,l-1}\Bar{R}_{l-1}^{-1}\Bar{B}_{2,l-1}^\mathsf{T})^{-1}\Bar{A}_{l-1}
\end{aligned}
\end{equation}
is stable, we show that stability holds for step $l$ as well.

\textit{Step $l$:} In this step, one can rewrite the closed-loop system in~\eqref{eq21} as $X_{l}^{*}=K_{X_l}x_{k|k}$, where $X_{l}^{*}=\col\{x_{k+1|k}^{*},\mathscr{X}_{k}^{*}\}$ and $K_{X_l}=\col\{K_{x_{k+1|k}},\mathscr{K}_{X_{l}}\}$ with $\mathscr{X}_{k}=\col\{x_{k+j|k}\}_{j=2}^l$ and $\mathscr{K}_{X_{l}}=\col\{K_{x_{k+j|k}}\}_{j=2}^l$. Again, by utilizing the Woodbury matrix identity and block matrix inversion, we obtain
\begin{equation}\label{eq37}
\begin{aligned}
\begin{bmatrix}x_{k+1|k}^{*} \\ \mathscr{X}_{k}^{*}\end{bmatrix}=\begin{bmatrix} P_{k+l}^{-1}(P_{k+l}^{-1}+BR^{-1}B^\mathsf{T})^{-1}A \\ \mathscr{K}_{X_{l}}P_{k+l}^{-1}(P_{k+l}^{-1}+BR^{-1}B^\mathsf{T})^{-1}A\end{bmatrix}x_{k|k},
\end{aligned}
\end{equation}
where
\begin{equation}\label{eq38}
\begin{aligned}
\mathscr{K}_{X_{l}}=\mathscr{Q}_{l}\Bar{B}_{1,l-1}^\mathsf{T}(&\Bar{B}_{1,l-1}\mathscr{Q}_{l}^{-1}\Bar{B}_{1,l-1}^\mathsf{T}\\
                &+\Bar{B}_{2,l-1}\Bar{R}_{l-1}^{-1}\Bar{B}_{2,l-1}^\mathsf{T})^{-1}\Bar{A}_{l-1}.
\end{aligned}
\end{equation}
The first equation in~\eqref{eq37} follows the same structure as the closed-loop system for $l = 1$, but with $P_{k+l}$ instead of $P_{k+1}$. According to Theorem~\ref{thm2}, step $l$, the Riccati equation has a steady-state PD solution, $P_{l} \succ 0$. By analogy to step 1, this equation is stable. Furthermore, the gain matrix in the second equation of \eqref{eq37} is the product of the stable gain from the first equation and another gain, assumed stable from step $l-1$, when replacing $P_{l}$ with $P_{l-1} \succ 0$. Therefore, the eigenvalues of the resulting gain matrix must lie within the unit circle, confirming stability. This proves that the gain matrix $K_{X_l}$ and, consequently, the closed-loop system are stable.
\end{proof}

\section{Simulation Study}
In this section, we provide numerical examples to illustrate the theoretical concepts discussed in the previous sections\footnote{System setup: MATLAB R2024a, Ryzen 9, 32GB RAM, Win 11.}. 
\begin{example}\label{exm1}
Consider the LTI system in~\eqref{eq1} with
\begin{equation*}
\begin{aligned}
A &=\begin{bmatrix}
    0.9 & 0.2 \\ -0.4 & 0.8
\end{bmatrix}, B =\begin{bmatrix}
    0.1 \\ 0.05
\end{bmatrix}, Q =\begin{bmatrix}
    0.5 & -0.5\\-0.5 & 10
\end{bmatrix}, R = 1,
\end{aligned}
\end{equation*}
subject to the states and input constraints
\begin{equation*}
\begin{aligned}
-0.45 \leq x_{1,2} \leq 0.5, \quad -0.25 \leq u \leq 0.25,
\end{aligned}
\end{equation*}
initial conditions $x_0 =(0.5,-0.1)^\mathsf{T}$, and $P_{k+l}=Q$.
Based on the given constraints, we construct the linear inequality constraints in ~\eqref{eq2} by defining $F_x = \operatorname{col}\{I_2, -I_2\}$, $g_x = \operatorname{col} \{0.5, 0.5,0.45,0.45\}$, $F_u = \operatorname{col}\{1, -1\}$, and $g_u = 0.25 \times \mathbf{1}_2$. Let us conduct the simulation with parameters $t_f=50$, $l=2$, and a large penalty value $\mu=10^3$. Applying the algorithm, the following observations are made:
\begin{itemize}
\item As illustrated in Fig.~\ref{fig1}, Re-MPC outperforms classical MPC (C-MPC) in stabilizing $x_1$ and $x_2$ by applying more control effort in the initial steps due to the regulation term and a dynamically adjusted design matrix.
\item Our evaluations show that Re-MPC achieves approximately a $15\%$ improvement in regulation performance compared to C-MPC, as indicated by the mean-square errors (MSE) of $x_1$ and $x_2$ (see Table~\ref{tab1}). Table~\ref{tab1} further highlights the superiority of Re-MPC in reducing total cost compared to existing approaches, despite a slightly higher computational time due to continuously calculating the solution of the Riccati equation at each step and updating the design matrix.
\end{itemize}
\begin{figure}[htbp]
\centering
\includegraphics[width=0.98\columnwidth]{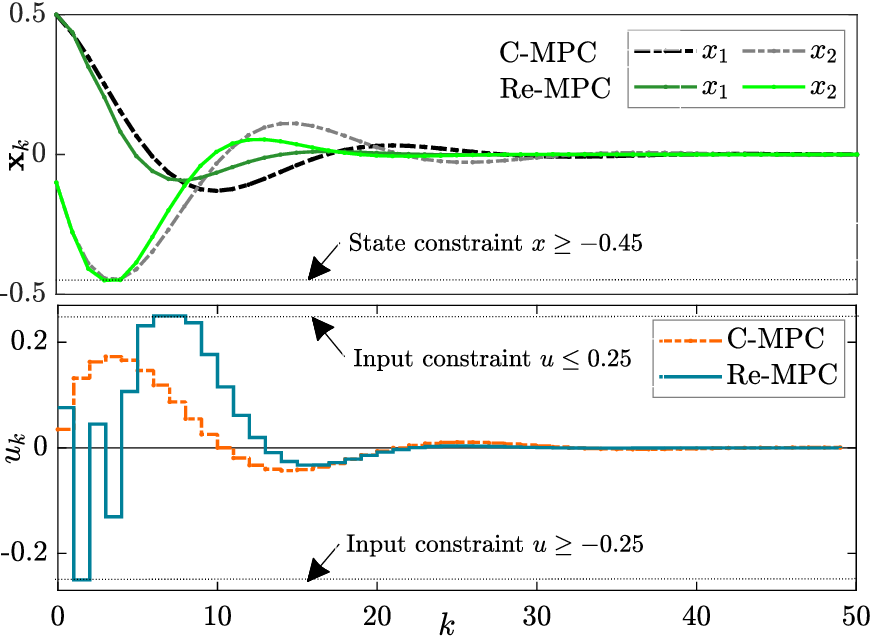}
\caption{Stabilized states and optimal control effort for C-MPC \& Re-MPC.}\label{fig1}
\end{figure}
\begin{table}[b]
\scriptsize
\caption{Performance Analysis I.}
\label{tab1}
\centering
\setlength{\tabcolsep}{3pt}
\begin{tabular}{lcccc}
\toprule
\centering
Method&  MSE (${x_1}$) & MSE (${x_2}$)& $\sum_{i=1}^{50} \|x_{i}\|_{Q}^2+\|u_{i}\|_{R}^2$ & Elapsed time (s) \\
\midrule
\centering
C-MPC~\cite{ref23} & $ 0.0146$ & $0.0218$ & $12.13$ & $2.2$\\
Proposed  & $0.0122$ & $0.0188$ & $10.66$ &  $3.43$\\
\bottomrule
\end{tabular}
\end{table}
\end{example}

\begin{example}\label{exm2} This example analyzes the effect of $\mu$ on the efficiency of Re-MPC through simulations using the same parameters as in the previous example, but with varying $\mu$.
\begin{itemize}
\item As seen in Fig.~\ref{fig2}, for small $\mu$, the algorithm behaves similarly to C-MPC. As it increases, performance improves and approaches the exact optimal solution at infinity (see Remark~\ref{rem1}). As seen, C-MPC serves as a lower bound, while Re-MPC with $\mu \to \infty$ acts as an upper bound for all possible solutions of the algorithm for varying $\mu$. 
\item As shown in Table 2, for small $\mu$, the relative changes (RCs) in the Euclidean norm of the design matrix of the proposed algorithm compared to C-MPC are negligible, and the total cost converges to that of C-MPC.
\end{itemize}

\begin{table}[htbp]
\scriptsize
\caption{Performance Analysis II.}
\label{tab2}
\centering
\begin{tabular}{lcccccc}
\toprule
   & $\mu=100$ & $\mu=50$ & $\mu=25$ & $\mu=10$ & $\mu=1$  \\
\midrule
$\sum_{i=1}^{50} \|x_{i}\|_{Q}^2+\|u_{i}\|_{R}^2$  &  $10.72$ & $10.87$ & $11.09$ & $11.44$ & $12.11$\\
RCs & 1.42 & 1.03 & 0.7 & 0.4 & 0.06\\
\bottomrule
\end{tabular}
\end{table}
\end{example}

\begin{figure}[htbp]
\centering
\includegraphics[width=0.98\columnwidth]{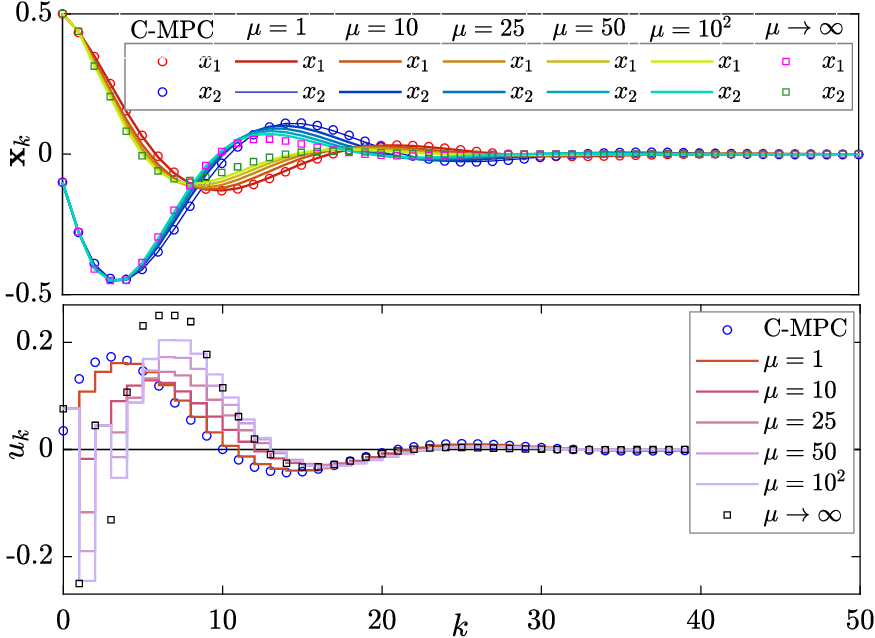}
\caption{Stabilized state and control effort for varying $\mu$.}\label{fig2}
\end{figure}

\section{Conclusions}
This work introduced a Riccati equation-based method for dynamically adjusting the design matrix in the classical MPCs. The proposed Re-MPC algorithm uses a recursive PLS-based Riccati equation to improve the performance of state regulation. Under controllability and detectability assumptions, we established convergence and stability guarantees at steady state. Numerical results demonstrate its superiority in balancing control effort and stability while reducing the total cost. In our method, as the penalty parameter approaches infinity, the algorithm converges to its best performance. This means that no parameter tuning is needed, which makes it perfect for online applications. Future work will extend Re-MPC to nonlinear and robust MPC designs for uncertain systems using its Riccati-based regulation.


\begin{thebibliography}{99}
\bibitem{ref1}
D. Saccani, G. F. Trecate, M. N. Zeilinger, and J. Köhler, ``Homothetic tube model predictive control
with multi-step predictors,'' \textit{IEEE Control Syst. Lett.}, vol. 7, pp. 3561--3566, 2023.

\bibitem{ref2}
J. Rawlings, D. Mayne, and M. Diehl, {\it{Model Predictive Control: Theory, Computation, and Design}}. Nob Hill Publishing, 2017.

\bibitem{ref3}
M. Lorenzen, M. Cannon, F. Allgöwer, ``Robust MPC with recursive model update,'' \textit{Automatica}, vol. 103, pp. 461--471, 2019.

\bibitem{ref4}
B. Kouvaritakis and M. Cannon, {\it{Model Predictive Control: Classical, Robust and Stochastic}}. Cham, Switzerland: Springer, 2016.

\bibitem{ref5}
C. E. Garcia, D. M. Prett, and M. Morari, ``Model predictive control: Theory and
practice-a survey,'' \textit{Automatica}, vol. 25, pp. 335--348, 1989.

\bibitem{ref6}
B. Zarrouki, M. Spanakakis, and J. Betz, ``A safe reinforcement learning driven weights-varying model predictive control for autonomous vehicle motion control,'' \textit{arXiv:2402.02624}, 2024.

\bibitem{ref7}
K, Shi, Z. Jiang, B. Liu, G. Yang, and M. Jin, ``Synergistic terrain-adaptive morphing and trajectory tracking in a transformable-wheeled roboT,'' \textit{IEEE Robot. Autom. Lett.}, vol. 10, no. 2, pp. 1656--1663, 2025.

\bibitem{ref23} F. Borrelli, A. Bemporad, and M. Morari, {\it{Predictive Control for Linear and Hybrid Systems}}. Cambridge University Press, 2017.

\bibitem{ref8}
A. Kapnopoulos and A. Alexandridis, ``A cooperative particle swarm optimization approach for tuning an mpc-based quadrotor trajectory tracking scheme,'' \textit{Aerosp. Sci. Technol.}, vol. 127, p. 107725, 2022.

\bibitem{ref9}
J. A. Paulson and A. Mesbah, ``Data-driven scenario optimization for automated controller tuning with probabilistic performance guarantees,'' in \textit{IEEE Control Syst. Lett.}, vol. 5, no. 4, pp. 1477--1482, 2021.

\bibitem{ref10}
G. Makrygiorgos et al., ``Performance-oriented model learning for control via multi-objective Bayesian optimization,'' \textit{Comput. Chem. Eng.}, vol. 162, p. 107770, 2022.

\bibitem{ref11}
M. Zanon and S. Gros, ``Safe reinforcement learning using robust MPC,'' \textit{IEEE Trans. Autom. Control}, vol. 66, pp. 3638--3652, 2021.

\bibitem{ref13}
W. Tang, ``Systematic MPC tuning with direct response shaping: Parameterization and Inverse optimization-based Tuning Approach (PITA),'' \textit{Control Eng. Pract.}, vol. 153, p. 106103, 2024.

\bibitem{ref14}
P. Bagheri and A. Khaki-Sedigh, ``An analytical tuning approach to multivariable
model predictive controllers,'' \textit{J. Process Control}, vol. 24, no. 12, pp. 41--54, 2014.

\bibitem{ref15}
Q. Lu, R. Kumar, and V. M. Zavala, ``MPC controller tuning using Bayesian optimization techniques,'' \textit{arXiv preprint:2009.14175}, 2020.


\bibitem{ref17}
A. S. Yamashita, A. C. Zanin, and D. Odloak, ``Tuning of model predictive control with multi-objective optimization,'' \textit{Braz. J. Chem. Eng.}, vol. 33, no. 2, pp. 333--346, 2016.

\bibitem{ref18}
M. Abtahi, M. Rabbani, and S. Nazari, ``An automatic tuning MPC with application to ecological cruise control,'' \textit{IFAC-PapersOnLine}, vol. 56, no. 3, pp. 265--270, 2023.

\bibitem{ref19}
P. J. Antsakalis and A. N. Michel, {\it{Linear Systems}}. McGraw Hill: 1998.

\bibitem{ref20} R. Nikoukhah, A. L. Willsky and B. C. Levy, ``{K}alman filtering and {R}iccati equations for decriptor systems,'' \emph{IEEE Trans. Autom. Control}., vol. 37, no. 9, pp. 1325--1342, 1992.

\bibitem{ref21} M. H. Terra, J. P. Cerri, and J. Y. Ishihara, ``Optimal robust linear quadratic regulator for systems subject to uncertainties,'' \emph{IEEE Trans. Autom. Control}.,  vol. 59, no. 9, pp. 2586--2591, 2014.

\bibitem{ref22} A. Bjorck, {\it{Numerical Methods for Least Squares Problems}}. Philadelphia, PA: SIAM, 1996.



\end{thebibliography}
\end{document}